\documentclass[preprint,aps,nofootinbib,superscriptaddress,11pt]{revtex4-1}
\usepackage{amssymb,mathrsfs,physics,amsthm,amssymb,amsfonts,amsmath,mathtools}
\usepackage[pdftex,pagebackref=false]{hyperref} 
\usepackage{hyperref}
\hypersetup{
    plainpages=false,       
    unicode=false,          
    pdftoolbar=true,        
    pdfmenubar=true,        
    pdffitwindow=false,     
    pdfstartview={FitH},    
    pdfnewwindow=true,      
    colorlinks=true,        
    linkcolor=Blue,         
    citecolor=Blue,        
    filecolor=magenta,      
    urlcolor=blue           
}
\usepackage[usenames, dvipsnames]{xcolor}

\usepackage{versions}
\excludeversion{cmt}

\usepackage{import}

\usepackage{enumitem}

\theoremstyle{definition}

\newtheorem{theorem}{Theorem}

\usepackage{upgreek}

\usepackage{graphicx}
\graphicspath{ {images/} }

\usepackage{chngcntr}
\usepackage{apptools}
\AtAppendix{\counterwithin{theorem}{section}}

\DeclarePairedDelimiterX{\inp}[2]{\langle}{\rangle}{#1, #2}

\newcommand\id{\leavevmode\hbox{\small1\kern-3.3pt\normalsize1}}

\newcommand{\cps}{\mathrel{\ooalign{$\nearrow$\cr \kern-0pt$\nwarrow$}}}

\usepackage{tikz}
\usetikzlibrary{arrows}
\usetikzlibrary{decorations.pathmorphing}

\usepackage{xparse}
\DeclareDocumentCommand{\mab}{m o}
  {%
    \mu\IfValueT{#2}{_{#2}}^{A\rightarrow B}{(#1)}
  }
\DeclareDocumentCommand{\mba}{m o}
  {%
    \mu\IfValueT{#2}{_{#2}}^{A\leftarrow B}{(#1)}
  }

\usepackage{CJKutf8}

\begin{document}

\title{Quantifying Causality in Quantum and General Models}

\begin{CJK*}{UTF8}{gbsn}
\author{Ding Jia (贾丁)}
\email{ding.jia@uwaterloo.ca}
\affiliation{Department of Applied Mathematics, University of Waterloo, Waterloo, Ontario, N2L 3G1, Canada}
\affiliation{Perimeter Institute for Theoretical Physics, Waterloo, Ontario, N2L 2Y5, Canada}

\begin{abstract}
In studies of entanglement, finding out if a state is entangled and quantifying the amount of entanglement contained in a state are related but different questions. Similarly in studies of causality, finding out the causal structures compatible with a model and quantifying the causal strengths are related but different questions. Recently much research have been directed towards the first question but considerably less attention is paid to the second one. In this paper we propose axioms for all reasonable quantitative measures of causality to obey. The axioms apply to a broad family of operational probabilistic theories with and without definite causal structure. For quantum models, we study causality measures based on one-shot quantum capacities in detail. These measures are used to define the notions of quantum signalling and quantum causality measures in order to quantify quantum causality.
\end{abstract}

\maketitle
\end{CJK*}

\begin{cmt}

Causality witness?

Q-factor?

\end{cmt}

\section{Introduction}

\begin{cmt}
causality in quantum models. causal discovery. indefinite causal structure. 

yes/no studied, strength not. focus of paper. coarse-grain out signalling correlation that are too noisy for practical purposes and for fundamental purposes (ICS). outline.
\end{cmt}

In operational probabilistic theories causality is usually characterized by the signalling criterion, which says that if one agent can change the measurement outcome probabilities of another agent by choosing different operations, then the first agent can causally influence the second.

The strength of the signalling criterion is that it offers a natural way to characterize causality that applies to not only explicit theories such as quantum and classical theories but also all operational probabilistic theories in general. Nonetheless the signalling criterion is limited in that it only offers yes or no answers to questions about causal structure but does not \textit{quantify} causal strength.

As we seek deeper understandings of causality, the need to quantify causality naturally arises. For instance, in the studies of quantum causality related to quantum gravity causal fluctuations are expected to be a generic phenomenon induced by spontaneous fluctuations of quantum gravitational degrees of freedom \cite{jia2017quantum}. In quantum spacetime because of the universally present quantum fluctuation of causal structure, generically any two parties will have a finite probability of being causally connected if one uses the signalling criterion. This may sound peculiar at first, but it is actually analogous to familiar features of quantum theory such as quantum tunneling which indicates a finite probability for seemingly peculiar events to happen. That such causal fluctuations do not violate locality is explained in \cite{jia2018analogue}, and if we accept what the theory suggests then quantum spacetime has a trivial causal structure at the yes or no level characterization of causality. Clearly the problem is that the signalling criterion does not distinguish ``strong'' and ``weak'' causal connections, and if only one raises the causal strength threshold of qualification towards causal connection by a little bit, most of the previous causal connections will be disqualified and the causal structure will become non-trivial. To apply this idea concretely one needs to study quantitative measures of causal strength.

An analogy can be drawn with entanglement theory. Although earlier studies focus on yes/no criterion for whether some parties share entanglement, many important questions were addressed only through studying quantitative measures of entanglement. 

In this paper we define and study causality measures for operational probabilistic theories with definite causal structure and with indefinite causal structure. The definition consists of three axioms for all reasonable causality measures to obey. For quantum theory we study in some detail two particular causality measures based on one-shot communication capacities. We show that for a family of important models describing indefinite causal structure, the one-shot entanglement transmission capacities are exactly solvable. We prove that the the one-shot entanglement transmission capacities can actually be used to reconstruct the causally relevant part of the models themselves.

For quantum theory there are correlations (e.g., the classical identity channel) that allow the transmission of only classical information but not quantum information. In some contexts such as the study of quantum spacetime there is the need to exclude these correlations from having positive causal strength. For this purpose we introduce the notion of ``quantum signalling'' and use it to define ``quantum causality measures'' that fits the purpose.

The present work focuses on studying quantifying causality between two parties. We leave the task of generalizing to multiple parties for future work.

\section{Causality measures}\label{sec:cm}

In this section we list the axioms for causality measures and give some examples of causality measures. There are different frameworks for operational probabilistic theories with definite causal structure and with indefinite causal structure (e.g., \cite{hardy2005probability, leifer2006quantum, leifer2013towards, gutoski2007toward, chiribella2009theoretical, aharonov2009multiple, abramsky2009categorical, coecke2010quantum, hardy2011reformulating, hardy2012operator, chiribella2013quantum, oreshkov2012quantum, araujo2015witnessing, oreshkov2016causal,  oreshkov2015operational, oreshkov2016operational, oeckl2016local, ried2015quantum, maclean2017quantum, fitzsimons2015quantum, costa2016quantum, allen2017quantum}). The following definition of causality measures applies to a wide range of frameworks. The only preliminary concepts needed are correlations (such as a channel) that mediate the causal influence, and local operations that change the correlations in order to exert the causal influence. 

Notably, nothing restricts the definition to quantum theory and everything in this section applies to any operational probabilistic theory with these preliminary concepts. It is only starting with the next section that we focus on quantum theory which allows us to talk about maximally entangled states and the fidelity of states to study some particular causality measures. 

Another point worth emphasizing is that the frameworks do not have to be based on directed acyclic graphs (DAGs) for the causality measures to be applicable. This is important because although many causal frameworks are based on DAGs, there are reasonable frameworks that are naturally associated with hypergraphs rather than graphs, such as Hardy's causaloid framework of indefinite causal structure \cite{hardy2005probability}. At the level of our current study of bipartite causality measures this general applicability is not significant, but it may prove to be advantageous in future works that generalize the study of causality measures to multiple parties.

\subsection{Axioms}\label{subsec:a}

A \textbf{causality measure} $\mu^{A\rightarrow B}(G)$ on parties $A$ and $B$ sharing the correlation $G$ is a real-valued function obeying the following axioms:
\begin{enumerate}
\item $\mu^{A\rightarrow B}(G)$ is non-increasing under local operations within $A$ and $B$.
\item $\mu^{A\rightarrow B}(G) \ge 0$.
\item $\mu^{A\rightarrow B}(G) > 0$ only if $A$ can signal to $B$ using $G$.
\end{enumerate}
Here ``$A$ can signal to $B$ using $G$'' means that by exploiting the correlation $G$, $A$ can change the measurement outcome probabilities of $B$ by choosing different operations. A \textbf{normalized causality measure} further obeys $\sup_G \mu^{A\rightarrow B}(G)=1$ so that $0\le \mu^{A\rightarrow B}(G)\le 1$ for all $G$. The causality measure $\mu^{A\leftarrow B}(G)$ in the opposite direction is defined similarly except that it obeys Axiom 3 with $A$ and $B$ swapped.

Axiom 1 is the main axiom for causality measures. It captures the intuition that the local operations cannot generate causal correlations. An arbitrary $G$ can be mapped to any correlation $G'$ that can be prepared by local operations alone (such as product states). The parties simply discard $G$ and prepare $G'$. Axiom 1 implies that $\alpha=\mu^{A\rightarrow B}(G')$ is the minimum value $\mu^{A\rightarrow B}$ can reach for all $G$, because starting from any $G$ the parties can apply local operations to prepare $G'$. Axiom 1 also implies that any two different $G'$ must share the same value of $\alpha$ for $\mu^{A\rightarrow B}$, because each can be prepared from the other. Axiom 2 sets this minimum value $\alpha$ to zero.

Axioms 1 and 2 resemble the axioms for entanglement measures \cite{horodecki2009quantum}, which was originally defined for states and recently generalized to general quantum correlations including those with indefinite causal structure \cite{jia2017generalizing}. The defining axioms of entanglement measures are that the measures do not increase under the LOCC (local operations and classical communications) operations, and that the measures are non-negative. More precisely, the first axiom for entanglement measures says that they should not increase under LOCC operations allowed by the LOCC setting that one is considering (monotonicity). Here an LOCC setting dictates what LOCC operations are allowed. For example, in some LOCC settings only one-way classical communication is allowed, and in some others no classical communication is allowed. The only difference between the entanglement measure axioms and Axioms 1 and 2 above is that entanglement measures must also be monotonic in the presence of classical communications if the LOCC setting allows them. In LOCC settings where all local operations are allowed (which is the case for most LOCC settings of interest), the monotonicity axiom of entanglement is stronger than Axiom 1 for causality measures. Therefore in a framework\footnote{Although as stated in \cite{jia2017generalizing} entanglement measures are defined specifically for quantum theory, they can easily be generalized to apply to a broad family of probabilistic theories which supports the notion of LOCC operations.} where they are defined the entanglement measures obey Axioms 1 and 2. However, a correlation that contains entanglement certainly does not necessarily contain causal correlation. Therefore Axiom 3 is needed to make sure that causality measures indeed measure causality. Incidentally, in a model where entanglement and causality measures are defined, if the LOCC setting only allows local operations, then causality measures obey the entanglement measure axioms. One could view causality measures as special cases of entanglement measures which obey Axiom 3 in the LOCC setting without classical communication.

We believe the remarks above justify the three axioms as necessary to define causality measures. There remains the question of whether more axioms are needed. One obvious option is to strengthen Axiom 3 by also requiring that $\mu^{A\rightarrow B}(G)>0$ if $A$ can signal to $B$ using $G$. We do not to make this requirement because it exclude some useful information transmission capacities as causality measures. For example, there are channels that can signal but have zero quantum channel capacity.

Another option is to require $\sum_i p_i \mab{G_i} \ge \mab{\sum_i p_iG_i}$ for probability vectors $p_i$. We do not to make this convexity requirement because again it would exclude quantum channel capacity as a causality measure \cite{smith2008quantum}. This choice echoes the choice in entanglement theory not to require entanglement measures to be convex (some useful measures such as distillable entanglement are not know to be convex) \cite{horodecki2009quantum}.

There are potentially other conditions one may want to impose on causality measures, just like there are conditions one may want to impose on entanglement measures in addition to the basic monotonicity and non-negativity axioms. For entanglement theory, the common view is that the two axioms above are the only ones necessary in defining entanglement measures, and other conditions may be imposed depending on particular contexts \cite{horodecki2009quantum}. It seems the case is the same for causality measures and we regard axioms 1 to 3 as sufficient to define causality measures at the basic level. Other conditions may be imposed to suit particular interests. For example, in Section \ref{sec:qcm} we study the the additional condition based on ``quantum signalling'' to define quantum causality measures.

In addition to general causality measures, we also defined normalized causality measures $\mu^{A\rightarrow B}$ for which $\sup_G \mu^{A\rightarrow B}(G)=1$. Normalized measures are useful when one compares correlations for systems with different dimensions. For example, the qubit identity channel and the qutrit identity channel are both channels with no noise and with the maximum causal strength on their respective systems. Yet the quantum channel capacity as a standard causality measure assigns a larger value to the qutrit channel. This assignment is reasonable from the perspective that the qutrit channel is capable of transmitting more information per use. Nevertheless, in other contexts where one quantifies causal strength according to how much noise there is in the correlation, a normalized measure that assigns the value one to both channels would be preferable.

\subsection{Examples}

\begin{itemize}

\item The zero measure.
\begin{align}
\mu_{\text{zero}}^{A\rightarrow B}(G)=0 \quad \text{for all }G.
\end{align}
This function trivially obeys all the three axioms and also the axioms for entanglement measures. It is of no practical value but shows that some function is both a causality measure and an entanglement measure.

\item The signalling measure.
\begin{align}
\mu_{\text{sg}}^{A\rightarrow B}(G)=
\begin{cases}
1, \quad \text{A can signal to B}
\\
0, \quad \text{A cannot signal to B}.
\end{cases}
\end{align}
This function clearly obeys Axioms 1 to 3 and is a causality measure. It is also a normalized causality measure.

Better than, for example, quantum channel capacity, it meets the condition that $\mu^{A\rightarrow B}(G)>0$ if $A$ can signal to $B$ using $G$. Yet it is not convex. Let $G_1$ and $G_2$ be channels that can and cannot signal. Then
\begin{align}
\frac{1}{2} \mu_{\text{sg}}^{A\rightarrow B}(G_1)+\frac{1}{2} \mu_{\text{sg}}^{A\rightarrow B}(G_2)=\frac{1}{2} < 1= \mu_{\text{sg}}^{A\rightarrow B}(\frac{1}{2}G_1+\frac{1}{2}G_2).
\end{align} The biggest drawback is that the signalling measure does not really \textit{quantify} causal strength.

\item For the special case that $G$ is a quantum or classical channel, the various channel capacities are causality measures, as one can easily check that they obey all the three axioms.

The channel capacities quantify how many qubits the channel can transmit per use and are not normalized measures in general. One can easily normalize them by dividing the maximum capacity a channel on the same input and output systems can reach. Precisely, given the channel capacity $C(N)$ on channels $N$ as a causality measure, we normalize it by
\begin{align}\label{eq:ncc}
C_{\text{norm}}(N):=\frac{C(N)}{\sup_{N'\in \mathfrak{C}(N)}C(N')},
\end{align}
where $\mathfrak{C}(N)$ is the set of channels on the same input and output systems of $N$. This suits the need mentioned at the end of the last subsection of finding a measure that assigns the same value one to all noiseless channels.

\item In the next section we study information transmission capacities that can be used as causality measures for general correlations not restricted to channels. They apply even to correlations without definite causal structure. The primary examples are the one-shot entanglement transmission capacities $Q_{\text{ent}}^\rightarrow(G^{AB};\epsilon)$ and the one-shot subspace transmission capacities $Q_{\text{sub}}^\rightarrow(G^{AB};\epsilon)$.

\item Given any causality measure $\mu^{A\rightarrow B}(G)$, a standard way to define a normalized measure is
\begin{align}\label{eq:ncm}
\mu^{A\rightarrow B}_{\text{norm}}(G):=\frac{\mu^{A\rightarrow B}(G)}{\sup_{\mathfrak{C}(G')}\mu^{A\rightarrow B}(G')},
\end{align}
where the supremum is taken over a set of correlations $\mathfrak{C}(G')$ that depends on $G'$. The normalization (\ref{eq:ncc}) for channel capacities is a special case of this general procedure.
\end{itemize}

\section{One-shot quantum capacities}\label{sec:osqc}

The channel capacities provide fairly natural quantitative measures of the causal strength. Yet the traditional definitions of capacities only apply to channels, which are correlations with definite causal structure. The various definitions of   capacities can be generalized to apply to correlations with possibly indefinite causal structure \cite{jiasakharwade}. These generalized capacities can then be used to quantify causality for e.g., quantum gravity, where indefinite causal structure is important.

In this section we focus on two canonical causality measures based on one-shot quantum communication capacities. In communication theory, the asymptotic capacities are usually taught as the canonical capacities \cite{wilde2017quantum}. Yet there is more than one reason to consider the one-shot capacities are the truly canonical ones. The conceptual reason is that practically all correlations for communication comes with noise, so the copies of correlations cannot be strictly identical. In addition, the copies may correlation with each other. Moreover, there is no supply of infinitely many copies of the correlation. The asymptotic capacities do not account for these practical limitations but the one-shot capacities do. The technical reason for preferring one-shot capacities over asymptotic ones is that the asymptotic capacities can be viewed as special cases of the one-shot capacities when the correlation used for communication is a tensor product of $n$ identical correlations and in the limits $n\rightarrow \infty$ and $\epsilon\rightarrow 0$ where $\epsilon$ is the error tolerance. An added reason from indefinite causal structure is that its presence the asymptotic capacities cannot be defined in the most straightforward way \cite{jia2017process}. At present it is still an open question what the best way to define asymptotic capacity is in the presence of indefinite causal structure, but the one-shot capacities do not suffer the same issue.

The one-shot communication tasks of entanglement transmission and subspace transmission and their capacities are originally defined for channels in \cite{buscemi2010quantum}. We generalize the previous definitions to incorporate communication resources with indefinite causal structure. In the next section we solve for the values of the capacities for some simple but important models of indefinite causal structure.

As mentioned, the notion of causality measures applies to frameworks more general than quantum ones and relies only on the concepts correlations that mediate the causal influence, and local operations that can change the correlations in order to exert causal influence. The following tasks require in addition that the states to be transmitted live on complex Hilbert spaces, and that the allowed local operations contain preparations of maximally entangled states for the entanglement transmission task, and preparations of arbitrary pure states on subspaces for the subspace transmission tasks. 

\subsection{Entanglement transmission capacity}\label{subsec:etc}

The goal of the entanglement transmission task is to transmit locally prepared entanglement through the correlation into shared entanglement. Suppose $A$ and $B$ share a correlation $G^{AB}$ that allows $A$ to send quantum states on Hilbert spaces of at most dimension $\tilde{m}$. In the $A$ to $B$ one-shot entanglement transmission task for the correlation $G^{AB}$, for each given dimension $m\le \tilde{m}$, $A$ first prepares a maximally entangled state $\Upphi^{MM'}\in L(\mathcal{H}^M\otimes \mathcal{H}^{M'})$ ($L(\mathcal{H}^x)$ denotes bounded linear operators on Hilbert space $\mathcal{H}^x$) with $\mathcal{H}^M\subset \mathcal{H}^A$, where $\mathcal{H}^A$ is the largest system $A$ can prepare states on, $\dim \mathcal{H}^M=m$, and $M'$ is a copy of the system $M$. Then $A$ keeps the $M$ part of the state intact to herself and send the $M'$ part of the state to $B$ using $G^{AB}$ such that they share a state $\Psi^{MM'}(E,D)$ with part $M$ held by $A$ and part $M'$ held by $B$. In this transmission $A$ applies some encoding local operation $E$ (which must keep the $M$ part of the original state $\Upphi^{MM'}$ intact) and $B$ applies some decoding local operation $D$. The goal is for $\Psi^{MM'}(E,D)$ to be as close to $\Upphi^{MM'}$ as possible.

Here we make a distinction between \textit{active} entanglement transmission and \textit{passive} entanglement transmission. The task of passive entanglement transmission is just as stated above, for which $A$ must keep the system $M$ intact after the original maximally entangled state $\Upphi^{MM'}$ is prepared. The task of active entanglement transmission on the other hand allows $A$ to apply some local operation $E'$ on the $M$ part of $\Psi^{MM'}(E,D)$ to obtain $\Psi^{MM'}(E,E',D)$ before it is compared with the target state $\Upphi^{MM'}$.\footnote{We require that $A$ must finally share entanglement with $B$ by keeping the $M$ part of the originally prepared state in order not to confuse the task of entanglement transmission with entanglement generation. Otherwise the parties can use an entangled state $G^{AB}$ that does not allow signalling to set up shared entanglement by simply discarding the originally prepared $\Upphi^{MM'}$ and keeping $G^{AB}$.} In the literature the distinction between the active and passive entanglement transmission tasks is often not stated explicitly, with some articles adopting the former (e.g., \cite{buscemi2010quantum}) and some others (e.g., \cite{tomamichel2016quantum}) adopting the latter as the ``entanglement transmission task''. It is not clear to us whether the two tasks are equivalent (having the same capacity), so we prefer to state them as different tasks.

For any positive integer $m\le \tilde{m}$, define the $A$ to $B$ \textbf{entanglement transmission fidelity} for $G^{AB}$ as:
\begin{align}\label{eq:gef1}
F_{\text{ent}}(G^{AB};m):=&\max_{\substack{\mathcal{H}^M\subset\mathcal{H}^A\\\dim \mathcal{H}^M=m}}\max_{E,D} \bra{\Upphi^{MM'}}\Psi^{MM'}(E,D)\ket{\Upphi^{MM'}}\quad \text{for the passive task},
\\
F_{\text{ent}}(G^{AB};m):=&\max_{\substack{\mathcal{H}^M\subset\mathcal{H}^A\\\dim \mathcal{H}^M=m}}\max_{E,E',D} \bra{\Upphi^{MM'}}\Psi^{MM'}(E,E',D)\ket{\Upphi^{MM'}}\quad \text{for the active task},\label{eq:gef2}
\end{align}
where $\ket{\Upphi^{A\bar{A}}}\in \mathcal{H}^A\otimes \mathcal{H}^{\bar{A}}$ is the pure state corresponding to $\Upphi^{A\bar{A}}$. In the first maximization the parties try over all the encodings and decodings. In the second maximization $A$ tries over all the subspaces.
Let $0\le \epsilon\le 1$ be a real number. $R=\log m$ is an \textbf{$\epsilon$-achievable rate} if 
\begin{align}
F_{\text{ent}}(G^{AB};m)\ge 1-\epsilon.
\end{align}
The $A$ to $B$ \textbf{one-shot entanglement transmission capacities} of $G^{AB}$ are defined as
\begin{align}
Q_{\text{ent}}^\rightarrow(G^{AB};\epsilon):=\max\{R: R \text{ is $\epsilon$-achievable}\}.
\end{align}
The tasks of $B$ to $A$ transmission with capacity $Q_{\text{ent}}^{\leftarrow}(G^{AB};\epsilon)$ can be defined analogously.

\subsection{Subspace transmission capacity}

The definitions for subspace transmission are analogous. The goal is to transmit any state in some subspace with high fidelity. Suppose $A$ and $B$ share a correlation $G^{AB}$ that allows $A$ to send quantum states on Hilbert spaces of at most dimension $\tilde{m}$. In the $A$ to $B$ transmission task, for each $m\le \tilde{m}$, $A$ picks a subspace $\mathcal{H}^M\subset \mathcal{H}^A$ where $\mathcal{H}^A$ is the largest system $A$ can prepare states on and $\dim \mathcal{H}^M=m$. Arbitrary pure states $\ket{\psi}\in \mathcal{H}^M$ are sent through $G^{AB}$ from $A$ to $B$ such that in the end $B$ gets a state with density operator $\Psi(E,D)\in L(\mathcal{H}^M)$. In the transmission $A$ applies some encoding local operation $E$ and $B$ applies some decoding local operation $D$. The goal is for $\Psi(E,D)$ to be as close to $\ketbra{\psi}$ as possible.

For any positive integer $m\le \tilde{m}$, define the \textbf{minimum output fidelity} for $G^{AB}$ as:
\begin{align}
F_{\text{min}}(G^{AB};m):=\max_{\substack{\mathcal{H}^M\subset \mathcal{H}^A\\ \dim \mathcal{H}^M=m}}\max_{E,D}\min_{\ket{\psi}\in \mathcal{H}^M}\bra{\psi}\Psi(E,D)\ket{\psi}.
\end{align}
In the first maximization the parties try over all the encodings and decodings. In the second maximization $A$ tries over all the subspaces. Let $0\le \epsilon\le 1$ be a real number. $R=\log m$ is an \textbf{$\epsilon$-achievable rate} if 
\begin{align}
F_{\min}(G^{AB};m)\ge 1-\epsilon.
\end{align}
The \textbf{one-shot subspace transmission capacities} of $G^{AB}$ are defined as
\begin{align}
Q_{\text{sub}}^\rightarrow(G^{AB};\epsilon):=\max\{R: R \text{ is $\epsilon$-achievable}\}.
\end{align}
The task for $B$ to $A$ transmission with capacity $Q_{\text{sub}}^{\leftarrow}(G^{AB};\epsilon)$ can be defined analogously.

\section{Capacities for simple models with indefinite causal structure}

In this section we solve for the values of the one-shot capacities for some simple but important models of indefinite causal structure. The models are defined in the process matrix framework \cite{oreshkov2012quantum}, which we lighteningly review below, referring the readers to the original article for details.

\subsection{Process matrices}\label{subsec:pm}

The process matrices are introduced by Oreshkov, Costa and Brukner \cite{oreshkov2012quantum} to incorporate indefinite causal structure into quantum theory. The main idea is to assume that ordinary quantum theory with definite causal structure holds locally, while globally the causal structure can be indefinite.

The local parties where ordinary quantum theory with definite causal structure holds are denote by $A, B, \cdots$. Local parties are where local operations are applied. The correlations that mediate causal influence are the process matrices, usually denoted by $W$. At the fundamental level, the correlations allow one to derive probabilities of observational outcomes within local parties -- they are maps from the local observational outcomes to the real numbers. The process matrices are representations of such maps as operators in Hilbert spaces.

Let the outcomes $i\in\mathcal{I}$ of a local observation be represented by the Choi operators $\hat{M_i}$ of the elements of a quantum instrument $\{M_i\}_{i\in \mathcal{I}}$. When the correlation is described by the process matrix $W$, the probability of observing the joint outcome $(i,j,\cdots,k)$ of $i$ within $A$, $j$ within $B$, ..., $k$ within $C$ is
\begin{align}\label{eq:br}
p(i,j,\cdots,k)=\Tr[(\hat{M}_i\otimes \hat{N}_j\otimes\cdots\otimes \hat{L}_k)^T W],
\end{align}
where $T$ denotes transpose.

Recall that a quantum instrument element $M_i$ within $A$ as a CP map $M_i:L(\mathcal{H}^{a_1})\rightarrow L(\mathcal{H}^{a_2})$ is associated with an input Hilbert space $\mathcal{H}^{a_1}$ and an output Hilbert space $\mathcal{H}^{a_2}$. If $\mathcal{H}:=\mathcal{H}^{a_1}\otimes\mathcal{H}^{a_2}\otimes \mathcal{H}^{b_1}\otimes\mathcal{H}^{b_2} \otimes \cdots \otimes \mathcal{H}^{c_1}\otimes\mathcal{H}^{c_2}$, then $W\in L(\mathcal{H})$. Let $\abs{x}$ stand for the dimension of the Hilbert space $\mathcal{H}^x$. Then the requirements that probabilities are non-negative and normalized imply that $W$ is positive semi-definite, has trace equal to $\abs{a_2}\abs{b_2}\cdots \abs{c_2}$, and lives in a linear subspace of $L(\mathcal{H})$, the detail of which we will not need for this paper. Conversely, any operator in $L(\mathcal{H})$ obeying these three conditions is a process matrix associated with local parties $A,B,\cdots,C$. In order to make explicit the parties associated with a process matrix we sometimes write $W$ as $W^{AB\cdots C}$.

Not only are channels and states in the form of their Choi operators process matrices, but also ordinary quantum theory with definite causal structure is a subtheory within the process matrix framework. The process matrix framework is therefore a generalization of ordinary quantum theory (with finite dimensional systems).

\subsection{Simple models with indefinite causal structure}\label{sec:tm}

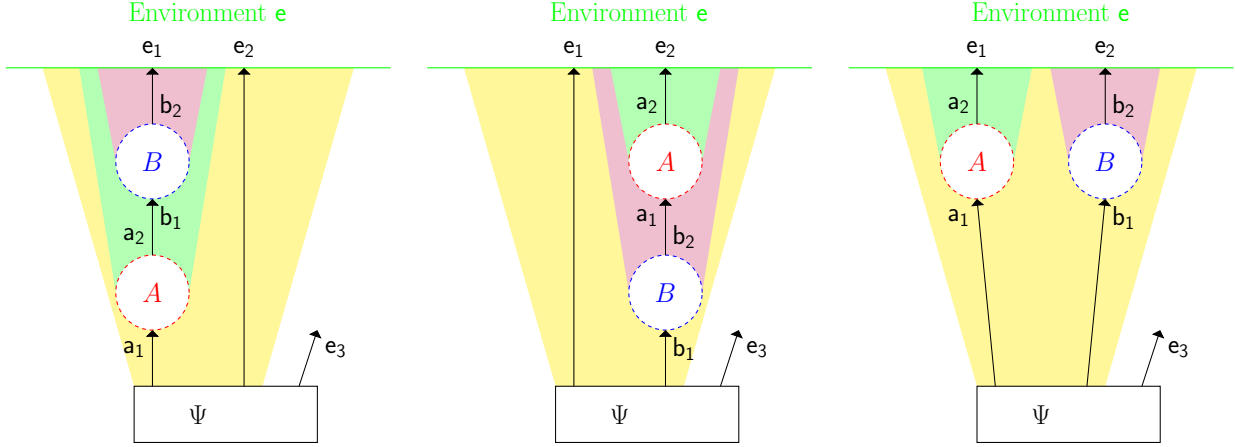
\begin{figure}
    \centering
    \resizebox{1\textwidth}{.36\textwidth}{
\begin{tikzpicture}

\fill [yellow!50] (11,-4.5)  -- (8.5,4) -- (17,4) -- (14.5,-4.5) -- cycle;
\fill [green!30] (10.5,-2)  -- (9.5,4) -- (13.5,4) -- (12.5,-2) -- cycle;
\fill [purple!25] (10.5,1.5)  -- (10,4) -- (13,4) -- (12.5,1.5) -- cycle;
\draw[thick,red,dashed, fill=white]  (11.5,-2) circle (1);
\node[red] at (11.5,-2) {\huge $A$};
\draw[thick,blue,dashed,fill=white]  (11.5,1.5)circle (1);
\node[blue]  at (11.5,1.5) {\huge $B$};
\draw[thick, green] (7.5,4) -- (18,4);
\node[green] at (13,5.5) {\huge Environment $\mathsf{e}$};
\node at (12.75,-5.25) {\huge $\Psi$};
\draw[thick]  (11,-4.5) rectangle (16,-6);
\draw[-triangle 90] [thick]  (14,-4.5) to (14,4);
\draw[-triangle 90] [thick]  (11.5,-4.5) to (11.5,-3);
\draw[-triangle 90] [thick]  (11.5,2.5) to (11.5,4);
\draw[-triangle 90] [thick]  (11.5,-1) to (11.5,0.5);

\node at (12,3) {\huge $\mathsf{b_2}$};
\node at (12,0) {\huge $\mathsf{b_1}$};
\node at (11,-0.5) {\huge $\mathsf{a_2}$};
\node at (11,-3.5) {\huge $\mathsf{a_1}$};

\fill [yellow!50] (34,-4.5)  -- (31.5,4) -- (40,4) -- (37.5,-4.5) -- cycle;
\fill [green!30] (33,1.5)  -- (32.5,4) -- (35.5,4) -- (35,1.5) -- cycle;
\fill [purple!25] (36.5,1.5)  -- (36,4) -- (39,4) -- (38.5,1.5) -- cycle;
\draw[thick,red,dashed,fill=white]  (34,1.5)circle (1);
\node[red] at (34,1.5) {\huge $A$};
\draw[thick,blue,dashed,fill=white]  (37.5,1.5)circle (1);
\node[blue] at (37.5,1.5) {\huge $B$};
\draw[thick, green] (30.5,4) -- (41,4);
\node[green] at (36,5.5) {\huge Environment $\mathsf{e}$};
\node at (35.75,-5.25) {\huge $\Psi$};
\draw[thick]  (34,-4.5) rectangle (39,-6);
\draw[-triangle 90] [thick]  (37,-4.5) to (37.5,0.5);
\draw[-triangle 90] [thick]  (34.5,-4.5) to (34,0.5);
\draw[-triangle 90] [thick]  (34,2.5) to (34,4);

\draw[-triangle 90] [thick]  (37.5,2.5) to (37.5,4);
\node at (33.5,3) {\huge $\mathsf{a_2}$};

\node at (38,0) {\huge $\mathsf{b_1}$};
\node at (38,3) {\huge $\mathsf{b_2}$};
\node at (33.5,0) {\huge $\mathsf{a_1}$};

\fill [yellow!50] (22.5,-4.5)  -- (20,4) -- (28.5,4) -- (26,-4.5) -- cycle;
\fill [purple!25] (24.5,-2)  -- (23.5,4) -- (27.5,4) -- (26.5,-2) -- cycle;
\fill [green!30] (24.5,1.5)  -- (24,4) -- (27,4) -- (26.5,1.5) -- cycle;

\draw[thick,red,dashed, fill=white]  (25.5,1.5) circle (1);
\node[red] at (25.5,1.5) {\huge $A$};
\draw[thick,blue,dashed,fill=white]  (25.5,-2)circle (1);
\node[blue] at (25.5,-2) {\huge $B$};

\draw[thick, green] (19,4) -- (29.5,4);
\node[green] at (24.5,5.5) {\huge Environment $\mathsf{e}$};
\node at (24.25,-5.25) {\huge $\Psi$};
\draw[thick]  (22.5,-4.5) rectangle (27.5,-6);

\draw[-triangle 90] [thick]  (25.5,-4.5) to (25.5,-3);
\draw[-triangle 90] [thick]  (23,-4.5) to (23,4);

\draw[-triangle 90] [thick]  (25.5,-1) to (25.5,0.5);
\draw[-triangle 90] [thick]  (25.5,2.5) to (25.5,4);

\node at (25,3) {\huge $\mathsf{a_2}$};
\node at (26,-3.5) {\huge $\mathsf{b_1}$};
\node at (26,-0.5) {\huge $\mathsf{b_2}$};
\node at (25,0) {\huge $\mathsf{a_1}$};
\node at (11.5,4.5) {\huge $\mathsf{e_1}$};
\node at (14,4.5) {\huge $\mathsf{e_2}$};
\node at (39.5,-3.5) {\huge $\mathsf{e_3}$};
\node at (34,4.5) {\huge $\mathsf{e_1}$};
\node at (37.5,4.5) {\huge $\mathsf{e_2}$};
\node at (23,4.5) {\huge $\mathsf{e_1}$};
\node at (25.5,4.5) {\huge $\mathsf{e_2}$};
\draw[-triangle 90] [thick] (38.5,-4.5) -- (39,-3);
\node at (28,-3.5) {\huge $\mathsf{e_3}$};
\draw[-triangle 90] [thick] (27,-4.5) -- (27.5,-3);
\node at (16.5,-3.5) {\huge $\mathsf{e_3}$};
\draw[-triangle 90] [thick] (15.5,-4.5) -- (16,-3);
\end{tikzpicture}
}
    
    \caption{The three causal relations. System $c$ not accessible to $A$ and $B$ is not drawn but is supposed to lie somewhere within the environment $e$.}
    \label{fig:3cr}
\end{figure}

To illustrate how a process matrix can encode indefinite causal structure, we consider a family of simple models. In the next subsection we will solve for the one-shot quantum capacities of these models. In a classical theory, two events $A$ and $B$ can have three possible causal relations: $A$ being in the causal past of $B$, $A$ being in the causal future of $B$, and $A$ being causally disconnected with $B$. When there is uncertainty about the causal relation there is indefinite causal structure. This can hold at the classical level, e.g., when the uncertainty is due to the lack of knowledge. A theory of gravitation that incorporates such uncertainty would need to be probabilistic \cite{hardy2016operational}. At the quantum level causal relations can in addition be indefinite in a quantum coherent way. We present a family of simple models that reflects the quantum coherent indefiniteness of the three causal relations between two parties. The model is used to represent spontaneous causal fluctuations of quantum spacetime. Such effects are assumed to be significant for all pairs of parties with a small separation at the order of the Planck length.

The \textit{harmonic clean models} \cite{jia2018analogue} (illustrated in FIG. \ref{fig:3cr}) put the three possible causal relations between parties $A$ and $B$ in ``superposition''. Let the complex numbers $\alpha_1,\alpha_2,\alpha_3$ be the probability amplitudes for the causal relations such that $\abs{\alpha_1}^2+\abs{\alpha_2}^2+\abs{\alpha_3}^2=1$. We introduce the shorthand notation $\alpha$ for $\alpha_1,\alpha_2,\alpha_3$ ($\alpha$ is like a complex 3-vector). Inspired by \cite{feix2017quantum} we define 
\begin{align}
\ket{w(\alpha)}^{GABE}:=&\alpha_1\ket{1}^g\ket{\Psi(\alpha)}^{a_1 e_2e_3}\ket{I}^{a_2 b_1}\ket{I}^{b_2 e_1}+\alpha_2\ket{2}^g\ket{\Psi(\alpha)}^{e_1 b_1e_3}\ket{I}^{b_2 a_1}\ket{I}^{a_2 e_2}\nonumber
\\
&+\alpha_3\ket{3}^g\ket{\Psi(\alpha)}^{a_1b_1e_3}\ket{I}^{a_2 e_1}\ket{I}^{b_2 e_2}.\label{eq:3crps}
\end{align}
We assume that the local subsystems of $A$ and $B$ all have equal dimensions, i.e., $\abs{a_1}=\abs{a_2}=\abs{b_1}=\abs{b_2}$. The vector $\ket{I}^{xy}$ (a pure maximally entangled state) is the Choi state representation of the identity channel from $x$ to $y$. The tripartite state $\ket{\Psi(\alpha)}$ (at the bottom of the pictures in FIG. \ref{fig:3cr}) is the ``initial state'' $A$ and $B$ can receive information from. In a more general model, both the channels between $A$ and $B$ and the state $\Psi$ may depend on $\alpha$. In the model we analyze we make the assumption that the channels do not. This simplification allows us to find exact answers for the causality measures later on, but a more general study may remove the assumption. Each of the three pictures of FIG. \ref{fig:3cr} depicts a definite causal structure between $A$ and $B$, and they correspond to the three terms of (\ref{eq:3crps}). The reason the systems are coupled through $\ket{\Psi}$ or $\ket{I}$ in the the particular way in (\ref{eq:3crps} can be inferred from the pictures in FIG. \ref{fig:3cr}.

The vector $\ket{w}$ takes the form of a superposition of three parts with amplitudes $\alpha_i$. It induces a four-party process matrix
\begin{align}
W^{GABE}(\alpha):=\ketbra{w(\alpha)}{w(\alpha)}^{GABE},
\end{align}
where $A$ consists of systems $a_1$ and $a_2$, $B$ consists of systems $b_1$ and $b_2$, $E$ consists of systems $e_1$, $e_2$ and $e_3$, and $G$ consists of the system $g$. For the local parties $A$ and $B$ the systems belonging to $E$ and $G$ are not accessible, so the process matrix reduces to
\begin{align}\label{eq:wcm}
W^{AB}(\alpha)=&\Tr_{GE}W^{GABE}(\alpha)=\sum_{i=1}^3 p_i W_i^{AB}(\alpha),
\\
\rho^{xye_3}=&\ketbra{\Psi(\alpha)}{\Psi(\alpha)}^{xye_3},
\\
\rho^{a_1}=&\rho^x, ~ \rho^{b_1}=\rho^y, ~ \rho^{a_1b_1}=\rho^{xy},\label{eq:rsc}
\\
W_1^{AB}(\alpha):=&\rho^{x}(\alpha)\otimes \Phi^{a_2b_1}\otimes \pi^{b_2}=\rho^{a_1}(\alpha)\otimes \Phi^{a_2b_1}\otimes \pi^{b_2},\label{eq:w1}
\\
W_2^{AB}(\alpha):=&\rho^{y}(\alpha)\otimes \Phi^{a_1b_2}\otimes \pi^{a_2}=\rho^{b_1}(\alpha)\otimes \Phi^{a_1b_2}\otimes \pi^{a_2},\label{eq:w2}
\\
W_3^{AB}(\alpha):=&\rho^{xy}(\alpha)\otimes \pi^{a_2}\otimes \pi^{b_2}=\rho^{a_1b_1}(\alpha)\otimes \pi^{a_2}\otimes \pi^{b_2}.\label{eq:w3}
\end{align}
The $\alpha$ dependence of $W_i$ comes from the $\alpha$ dependence of $\Psi$. The meaning of (\ref{eq:rsc}) is that $\rho^{a_1}$ is the reduced state of $\rho^{xye_3}$ on the first subsystem, $\rho^{b_1}$ is the reduced state of $\rho^{xye_3}$ on the second subsystem, and $\rho^{a_1b_1}$ is the reduced state of $\rho^{xye_3}$ on the first and second subsystems.

\subsection{Capacities}

The one-shot entanglement transmission capacities can be solved exactly for the harmonic clean models.
\begin{theorem}\label{th:osec}
For $W^{AB}$ in the family of harmonic clean models (\ref{eq:3crps}),
\begin{align}\label{eq:ec}
Q_{\text{ent}}^\rightarrow(W^{AB};\epsilon)=
\begin{cases}
\log\Big(\max\Big\{m\in \mathbb{N}: m\le \sqrt{\frac{1}{1-\frac{\epsilon}{1-p_1}}}\text{ and }m\le \abs{a_2}\Big\}\Big), & p_1<1-\epsilon,
\\
\log\abs{a_2}, & p_1\ge 1-\epsilon.
\end{cases}
\\\label{eq:ecl}
Q_{\text{ent}}^\leftarrow(W^{AB};\epsilon)=
\begin{cases}
\log\Big(\max\Big\{m\in \mathbb{N}: m\le \sqrt{\frac{1}{1-\frac{\epsilon}{1-p_2}}}\text{ and }m\le \abs{b_2}\Big\}\Big), & p_2<1-\epsilon,
\\
\log\abs{b_2}, & p_2\ge 1-\epsilon.
\end{cases}
\end{align}
\end{theorem}

\begin{proof}

We first find exact values of $F_{\text{ent}}(W^{AB};m)$ for the harmonic clean model for any $m\le \abs{a_2}$. Let $\Psi^{MM'}$ be the state $A$ and $B$ finally share and $\Upphi^{MM'}$ be the target maximally entangled state. $W^{AB}$ takes the form of a probabilistic mixture of $\sum_{i=1}^3 p_i W_i$ in (\ref{eq:wcm}), so $\Psi^{MM'}=:\sum_{i=1}^3 p_i\Psi^{MM'}_i$ is a probabilistic mixture of the three states $\Psi^{MM'}_i$ established through $W_i$. The function $F(\Psi^{MM'},\Upphi^{MM'})=\bra{\Upphi^{MM'}}\Psi^{MM'}\ket{\Upphi^{MM'}}$ is linear in the density operator $\Psi^{MM'}$, so $F(\Psi^{MM'},\Upphi^{MM'})=\sum_i p_i F(\Psi^{MM'}_i,\Upphi^{MM'})=:\sum_i p_i F_i$.

$W_1$ represents a noiseless channel, so it can achieve $F_1=1$ using a suitable protocol. In contrast, $F_2,F_3\le 1/m^2$. To see this, note from (\ref{eq:w2}) and (\ref{eq:w3}) that both $W_2$ and $W_3$ trace out $a_2$. Therefore the $M'$ part of the initial state $\Upphi^{MM'}$ is always eventually traced out. The final state takes the form $\Psi^{MM'}_i=\pi^{M}\otimes \rho_i^{M'}$ for $i=2,3$, where $\pi^M$ is the maximally mixed state and $\rho_i^{M'}$ are arbitray states. Then for $i=2,3$
\begin{align}
F_i=&\bra{\Upphi^{MM'}}\Psi^{MM'}_i\ket{\Upphi^{MM'}}
\\
=&\bra{\Upphi^{MM'}}\pi^{M}\otimes \rho_i^{M'}\ket{\Upphi^{MM'}}
\\
=&\frac{1}{m^2}\sum_k \bra{k}^{M'}\rho_i^{M'}\ket{k}^{M'}
\\
=&\frac{1}{m^2}\Tr_{M'}\rho_i^{M'}=\frac{1}{m^2}.
\end{align}

The above equation for $i=2,3$ hold for whatever $\rho_i^{M'}$ and hence for whatever protocol. Therefore a protocol is optimal for the entanglement transmission tasks if and only if it achieves $F_1=1$. An example is to input the $M'$ part of $\Upphi^{MM'}$ into $a_2$, trace out $a_1$ and $b_2$, and take $b_1$ directly as $M'$. Since $F_2=1/m^2$ and $F_3=1/m^2$, the the optimal fidelity in (\ref{eq:gef1}) or (\ref{eq:gef2}) for all $m\le \abs{a_2}$ takes the value
\begin{align}\label{eq:tmef}
F_{\text{ent}}(W^{AB};m)=\sum_i p_i F_i=p_1+(p_2+p_3)/m^2=p_1+(1-p_1)/m^2.
\end{align}

The entanglement fidelity (\ref{eq:tmef}), which holds for all $0<m\le \abs{a_2}$, is a monotonically increasing function of $p_1$. For fixed $\epsilon$ we can directly obtain the entanglement transmission capacity $Q_{\text{ent}}(W^{AB};\epsilon)$. A rate $0<m\le \abs{a_2}$ is achievable if and only if $F_{\text{ent}}(W^{AB};m)\ge 1-\epsilon$. Using  (\ref{eq:tmef}) this is equivalent to $1-F_{\text{ent}}(W^{AB};m)=(1-p_1)(1-1/m^2)\le \epsilon$. To proceed we need to compare $1-p_1$ and $\epsilon$. If $1-p_1>\epsilon$, a simple calculation reveals that $m\le \abs{a_2}$ is achievable if and only if
\begin{align}\label{eq:mub}
m\le \sqrt{\frac{1}{1-\frac{\epsilon}{1-p_1}}}.
\end{align}
If $1-p_1\le \epsilon$, all $0<m\le \abs{a_2}$ are achievable. The maximum achievable value, $\mathsf{m}$, gives the capacity $Q_{\text{ent}}^\rightarrow(W^{AB};\epsilon)=\log \mathsf{m}$. We have
\begin{align}
Q_{\text{ent}}^\rightarrow(W^{AB};\epsilon)=
\begin{cases}
\log\Big(\max\Big\{m\in \mathbb{N}: m\le \sqrt{\frac{1}{1-\frac{\epsilon}{1-p_1}}}\text{ and }m\le \abs{a_2}\Big\}\Big), & p_1<1-\epsilon,
\\
\log\abs{a_2}, & p_1\ge 1-\epsilon.
\end{cases}
\end{align}
A symmetric proof establishes (\ref{eq:ecl}).
\end{proof}

Clearly $Q_{\text{ent}}^\rightarrow(W^{AB};\epsilon)$ is a non-decreasing functions of $p_1$, which makes intuitive sense because as the amplitude of $A$ in the causal past of $B$ increases we expect a better capacity. For $A$ to $B$ communication, once $p_1$ is fixed the values of $p_2$ and $p_3$ are irrelevant. Note that the capacity also depends on $\abs{a_2}$, which imposes a bound on the maximum achievable value $\mathsf{m}$.

What relations of $\epsilon$ and $p_1$ give the extreme cases of the maximal capacity and the zero capacity? For fixed $\abs{a_2}$, it follows directly from (\ref{eq:ec}) that $Q_{\text{ent}}^\rightarrow(W^{AB};\epsilon)=\log\abs{a_2}$ if and only if \begin{align}
    \frac{\epsilon}{1-p_1}\ge 1-\frac{1}{\abs{a_2}^2}.
\end{align}
On the other hand, $Q_{\text{ent}}^\rightarrow(W^{AB};\epsilon)=0$ if and only if 
\begin{align}\label{eq:ccz}
p_1<1-\frac{4}{3}\epsilon.
\end{align}
In this case the RHS of (\ref{eq:mub}) is less than $2$, so $\mathsf{m}=1$ and $Q_{\text{ent}}^\rightarrow(W^{AB};\epsilon)=\log\mathsf{m}=0$.

Another important thing is that for any harmonic clean model $W$ with $p_1>0$, there is always an $\epsilon$ that allows it to quantum communicate better than non-signalling resources. For any non-signalling resource $V$, the same calculation for $F_{\text{ent}}$ of $W_2$ or $W_3$ applies, so $F_{\text{ent}}(V;m)=1/m^2$. For any non-trivial $m$ ($m\ge 2$) there is an $\epsilon>0$ such that
\begin{align}
p_1+(1-p_1)/m^2>1-\epsilon>1/m^2.
\end{align}
In particular, we can pick such an $\epsilon_0$ for $m=2$. The second inequality above implies that for this $\epsilon$, $Q_{\text{ent}}(V;\epsilon_0)=0$, and the first inequality above implies through (\ref{eq:tmef}) that $Q_{\text{ent}}(W;\epsilon_0)\ge \log 2$.

All the previous conclusions have a symmetric version for $B$ to $A$ communication. 

We move on to study subspace transmission. Let $\Psi$ be the state $B$ finally obtains. Similar to entanglement transmission, we observe that $F(\psi,\Psi)=\bra{\psi}\Psi\ket{\psi}$ is linear in $\Psi$. $\Psi$ in turn is a mixture of $\Psi_i$ established through $W_i$, so $F(\psi,\Psi)=\sum_i p_i F(\psi,\Psi_i)=:\sum_i p_i F_i$.

$W_1$ represents a noiseless channel, so it can achieve $F_1=1$ for any input state $\psi$ by sending it into $a_2$, tracing out $a_1$ and $b_2$, and taking the state that emerges at $b_1$ directly as $\Psi$. In contrast, once $E$ and $D$ are fixed $\Psi_i$ are constant states for $i=2$ and $i=3$, respectively. For these two cases, $\min_{\ket{\psi}} F_i=\min_{\ket{\psi}} \bra{\psi}\Psi_i\ket{\psi} =\lambda_{\text{min}}$, where $\lambda_{\text{min}}$ is the minimum eigenvalue of $\Psi_i$. A protocol that maximizes these values for $i=2,3$ outputs a maximally mixed state so that $\lambda_{\text{min}}=1/m$. Such a protocol always exists because $B$ can always trace out the state he gets and prepares a maximally mixed state, but the protocol may not be unique because $B$ may receive a maximally mixed state directly and hence does not need to re-prepare it.

If $\rho^{b_1}(\alpha_3)\neq \pi^{b_1}$, we cannot obtain a formula for the capacities without knowing $\rho^{b_1}(\alpha_3)$. Consider $F_{\text{min}}(W^{AB};m)$ for $m=\abs{a_2}$. A protocol that maximizes $\min_{\ket{\psi}} F_1$ must ask $B$ to take whatever that he receives at $b_1$ as the final state up to unitaries, because this is the only way that the information from the original state at $A$ reaches him without loss. Yet the same protocol applied to $W_i$ for $i=2,3$ would not maximize $\min_{\ket{\psi}} F_i$, since these would take  $\rho^{b_1}(\alpha_3)\neq \pi^{b_1}$ as the final state up to unitaries, and this state must have  $\lambda_{\text{min}}<1/m$, which is sub-optimal. The actual optimal protocol for the subspace transmission task will then have to make a compromise between optimizing $i=1$ and $i=2,3$, and the precise form of the protocol will depend on what $\rho^{b_1}(\alpha_3)$ is. 

If on the other hand $\rho^{b_1}(\alpha_3)= \pi^{b_1}$, then the same protocol that maximizes $\min_{\ket{\psi}} F_i$ for $i=1$ also maximizes it for $i=2,3$. In fact, this is the optimizing protocol for all $m$ not just $m=\abs{a_2}$. Therefore we have for all $m$, $F_{\text{min}}(W^{AB};m)=p_1+(p_2+p_3)/m=p_1+(1-p_1)/m$. A calculation similar to the one in the proof of Theorem \ref{th:osec} yields the following result.
\begin{theorem}\label{th:ossc}
For $W^{AB}$ in the family of harmonic clean models (\ref{eq:3crps}) with $\rho^{b_1}(\alpha_3)$ equalling the maximally mixed state $\pi^{b_1}$,
\begin{align}\label{eq:sc}
Q_{\text{sub}}^\rightarrow(W^{AB};\epsilon)=
\begin{cases}
\log\Big(\max\Big\{m\in \mathbb{N}: m\le \frac{1}{1-\frac{\epsilon}{1-p_1}}\text{ and }m\le \abs{a_2}\Big\}\Big), & p_1<1-\epsilon,
\\
\log\abs{a_2}, & p_1\ge 1-\epsilon.
\end{cases}
\\\label{eq:scl}
Q_{\text{sub}}^\leftarrow(W^{AB};\epsilon)=
\begin{cases}
\log\Big(\max\Big\{m\in \mathbb{N}: m\le \frac{1}{1-\frac{\epsilon}{1-p_2}}\text{ and }m\le \abs{b_2}\Big\}\Big), & p_2<1-\epsilon,
\\
\log\abs{b_2}, & p_2\ge 1-\epsilon.
\end{cases}
\end{align}
\end{theorem}

\subsection{The reconstruction theorem}

Incidentally, within the family of harmonic clean models the one-shot entanglement transmission capacities determines the causally relevant part of the harmonic clean models in the following sense.
\begin{theorem}[Reconstruction from One-Shot Quantum Capacities]\label{th:etr}
Let $W^{AB}$ be any unknown harmonic clean model. Then the capacities $Q_{\text{ent}}^\rightarrow(W^{AB};\epsilon)$ and $Q_{\text{ent}}^\leftarrow(W^{AB};\epsilon)$ indexed by $\epsilon$ determines $\abs{\alpha}:=(\abs{\alpha_1},\abs{\alpha_2},\abs{\alpha_3})$. If $\abs{\alpha_3}\neq 1$, they also determine the subsystem dimensions.
\end{theorem}
\begin{proof}
By definition, any particular harmonic clean model has $\abs{a_1}=\abs{a_2}=\abs{b_1}=\abs{b_2}$. We first show that if $\abs{\alpha_3}\neq 1$, the quantum capacities determine this dimension. It suffices to show that any pair of $W$ and $W'$ with different subsystem dimensions have a different capacity for some $\epsilon$. Without loss of generality assume that the subsystem dimension for $W$ is larger. By assumption $W$ has $\abs{\alpha_3}\neq 1$. Because $\abs{\alpha_3}\neq 1$ at least one of $p_1$ and $p_2$ is positive. Without loss of generality assume $p_1>0$. By (\ref{eq:ec}) for $\epsilon\ge 1-p_1$, $Q_{\text{ent}}^\rightarrow(W;\epsilon)=\log\abs{a_2}>\log\abs{a_2'}\ge Q_{\text{ent}}^\rightarrow(W';\epsilon)$. Hence the quantum capacities distinguish subsystem dimensions.

We next show that the capacities determine $\abs{\alpha}$. The capacities can tell if $\abs{\alpha_3}=1$, because by Theorem \ref{th:osec} only in this case all the capacities in both directions are zero. In this case $\abs{\alpha_1}=\abs{\alpha_2}=0$, so the capacities determine $\abs{\alpha}$. Next we assume that $\abs{\alpha_3}\neq 1$. We want to show that any pair of $W$ and $W'$ with the same subsystem dimension but different $\abs{\alpha}$ have different capacities for some $\epsilon$. If $W$ and $W'$ are different, then for $j=1$, $j=2$ or both, $\abs{\alpha_j}\neq \abs{\alpha'_j}$. Without loss of generality assume that $\abs{\alpha_1}< \abs{\alpha'_1}$ and hence $p_1< p'_1$. There is an $\epsilon>0$ such that 
\begin{align}
\frac{3}{4}(1-p_1)> \epsilon >\frac{3}{4}(1-p'_1).
\end{align}
By (\ref{eq:ccz}), $Q_{\text{ent}}^\rightarrow(W^{AB};\epsilon)=0$, and $Q_{\text{ent}}^\rightarrow(W'^{AB};\epsilon)> 0$. This establishes the theorem.
\end{proof}

This reconstruction theorem is potentially important in that it suggests a way to quantitatively characterize correlations with indefinite causal structure in general operational probabilistic theories beyond quantum models. The process matrices are defined over complex Hilbert spaces. Although correlations with indefinite causal structure had been defined in more general settings (e.g., Hardy's causaloid \cite{hardy2005probability} and Oreshkov and Giarmatzi's general processes \cite{oreshkov2016causal}), their quantitative features have not been studied in detail. The definition of the one-shot entanglement transmission capacities (and that of the one-shot subspace transmission capacities) may be extended to general models. They offer quantitative characterizations of the more general correlations and the reconstruction theorem suggests they may even characterize all the essential aspects of the correlation as far as the causality is concerned. An preliminary question that deserves further study is to what extent the one-shot entanglement transmission capacities can be used to reconstruct general quantum process matrices outside the family of harmonic clean ones.

\section{Quantum causality measures}\label{sec:qcm}

In operational probabilistic theories, signalling is commonly used as \textit{the} criterion for causality. Yet for quantum models (complex Hilbert space quantum operational probabilistic theories, for which ordinary quantum theory and quantum theories with indefinite causal structure are special cases) there are motivations to introduce another criterion for causality. Indeed, quantum and classical information are different types of information, and we know from communication theory that one may be able to transmit classical information without being able to transmit quantum information. The signalling criterion is defined with respect to influencing classical measurement outcomes, so it may be regarded as a causality criterion based on classical information. Is there a causality criterion based on quantum information?

We propose a quantum causality criterion based on the one-shot quantum transmission tasks defined in Section \ref{sec:osqc}. Roughly speaking the criterion says that if a correlation performs any of the one-shot quantum transmission task better than all non-signalling correlations for any error tolerance $\epsilon$, then the correlation can be used to ``quantum signal''. The traditional signalling criterion is weak in the sense that any influence of the measurement outcome probabilities qualifies a correlation to signalling. Similarly, the quantum signalling criterion is weak in the sense that a better-than-non-signalling-resource performance for any one-shot quantum transmission task qualifies a correlation to quantum signalling.

One use of the quantum causality criterion is to distinguish natural models of quantum spacetime which support indefinite causal structure from unnatural ones. The models of quantum spacetime that only support indefinite causal structure according to the signalling criterion are unnatural. When the medium of causal influence is some material such as a telephone line, it is conceivable that the material only allows the transmission of classical but not quantum information. However, for quantum spacetime itself as the medium, it would be very unnatural for two causally connected parties to share correlations that can only send classical but not quantum information. A natural model of quantum spacetime should have indefinite causal structure according to both the signalling and the quantum signalling criteria.

Another use of the quantum signalling criterion is to update the axioms of causality measures to define quantum causality measures that quantify quantum causal strengths (Subsection \ref{subsec:aqcm}). Quantum causality measures have applications, for instance, in quantifying the causal strength of quantum spacetime correlations.

\subsection{Quantum signalling}

Suppose $G$ is the quantum correlation the two parties $A$ and $B$ share. We say that $A$ can \textbf{quantum signal} to $B$ if there exists an error tolerance $\epsilon$ for which they can perform any one-shot quantum transmission task better than the non-signalling resources in the traditional sense. In other words, we say that $A$ can quantum signal to $B$ if there exists $\epsilon>0$ for which  $Q^{A\rightarrow B}(G;\epsilon)>\sup_{H\in \mathcal{N}}Q^{A\rightarrow B}(H;\epsilon)$, where $\mathcal{N}$ is the set of non-signalling resources defined on the same systems, and $Q$ is any of the one-shot quantum transmission capacities including the active and passive entanglement transmission capacities and the subspace transmission capacity. We call this the ``quantum signalling criterion''. To distinguish quantum signalling from the traditional notion of signalling, we call the latter ``classical signalling'', because it is defined based on classical observational outcomes.

Quantum signalling is stronger than than classical signalling, because by definition in order to quantum signal the parties must beat all classically non-signalling resources, which implies that they share a resource that is can classically signal.

Classical causal correlations do not allow quantum signalling. Classical correlations break entanglement and coherence. For the entanglement transmission task they can only set up shared separable states (otherwise entanglement may be created by LOCC) but not entangled states. Yet separable states will not have more entanglement fidelity than product states. Let $\sum_i p_i \rho^M_i\otimes \sigma^{M'}_i$ be an arbitrary separable state. Then
\begin{align}
F=\bra{\Upphi^{MM'}}\sum_i p_i \rho^M_i\otimes \sigma^{M'}_i\ket{\Upphi^{MM'}}=&\bra{\Upphi^{MM'}}\sum_i p_i \rho^M_i\otimes \sigma^{M'}_i\ket{\Upphi^{MM'}}
\\
=&\frac{1}{\abs{M}}\Tr_M\Big[\sum_i p_i \rho^M_i \sigma^{M}_i\Big]
\\
\le&\max_i \frac{1}{\abs{M}}\Tr_M\Big[\rho^M_i \sigma^{M}_i\Big]
\\
=&\max_i \frac{1}{\abs{M}}\bra{\Upphi^{MM'}} \rho^M_i\otimes \sigma^{M'}_i\ket{\Upphi^{MM'}}.
\end{align}
Therefore the separable state does not have greater entanglement fidelity than the product state $\rho^M_i\otimes \sigma^{M'}_i$ for some $i$. Because any product state can be created established without signalling, classical correlations do not perform better than the classically non-signalling resources for entanglement transmission. For the subspace transmission task, note that even the most effective classical causal correlation, the classical identity channel, cannot achieve a greater minimum output fidelity than classically non-signalling correlations. Suppose the classical identity channel projects onto the $\{\ket{i}\}_{i=1}^d$ basis. The most effective encodings and decodings are unitaries. Without loss of generality assume they are the quantum identity channels. The worst case scenario for the minimum output fidelity is with the input state $\ket{\psi}=\sum_{i=1}^d \frac{1}{\sqrt{d}} \ket{i}$. Then $\Psi=\frac{1}{d}\id$ and $F=\bra{\psi}\Psi\ket{\psi}=1/d$. This minimum output fidelity can be matched if $A$ and $B$ share a classically non-signalling correlation and for the transmission $B$ traces out whatever he receives and outputs the maximally mixed state. Therefore classical correlations do not perform better than the classically non-signalling resources for subspace transmission either and hence they cannot quantum signal.

We note that entanglement \textit{generation} capacities do not count as \textit{transmission} capacities, because as mentioned in Subsection \ref{subsec:etc} a correlation (e.g., an entangled state) that does not allow the transmission of quantum information may have a positive entanglement generation capacity. This situation contrasts that with the asymptotic capacities for quantum channels, for which the three capacities of entanglement transmission, subspace transmission, and entanglement generation agree. One reason for the difference is that restricted to channels nothing can generate entanglement without being able to transmit entanglement. On the other hand, correlations with indefinite causal structure such as process matrices contain entangled states as special cases, but these can generate entanglement without being able to transmit entanglement. The inclusion of such correlations break the ``degeneracy'' of the three capacities.

Another reason for the difference is that quantum signalling is defined using one-shot capacities rather than asymptotic ones. For quantum channels there is an inequality that relates the entanglement transmission and subspace transmission capacities \cite{buscemi2010quantum}:
\begin{align}\label{eq:oscieq}
Q_{\text{ent}}(N;\epsilon)-1\le Q_{\text{sub}}(N;2\epsilon)\le Q_{\text{ent}}(N;4\epsilon),
\end{align}
which shows that the two capacities are closely related. However, it does not set up an equivalence of the two capacities. Neither can it be used to pick one out of the two capacities to define quantum signalling to yield a weaker quantum signalling criterion than with the other one. The incomparability of the one-shot quantum capacities leaves us with the need to check each type of capacity to qualify quantum signalling.

\subsection{Axioms for quantum causality measures}\label{subsec:aqcm}

A \textbf{quantum causality measure} $\mu^{A\rightarrow B}(G)$ on local parties $A$ and $B$ sharing correlation $G$ is a real-valued function obeying the following axioms:
\begin{enumerate}
\item $\mu^{A\rightarrow B}(G)$ is non-increasing under local operations within $A$ and $B$.
\item $\mu^{A\rightarrow B}(G) \ge 0$.
\item $\mu^{A\rightarrow B}(G) > 0$ only if $A$ can quantum signal to $B$ using $G$.
\end{enumerate}
A \textbf{normalized quantum causality measure} further obeys $\sup_R \mu^{A\rightarrow B}(G)=1$ so that $0\le \mu^{A\rightarrow B}(G)\le 1$ for all $G$. The quantum causality measure $\mu^{A\leftarrow B}(G)$ in the opposite direction is defined similarly except that it obeys Axiom 3 with $A$ and $B$ swapped.

In comparison to causality measure axioms, the only difference is that in axiom 3 ``quantum signal'' is used in place of ``signal''.

\subsection{Examples of quantum causality measures}

\begin{itemize}

\item The zero measure.
\begin{align}
\mu_{\text{zero}}^{A\rightarrow B}(G)=0 \quad \text{for all }G.
\end{align}
This function trivially obey all the three axioms. It is a causality measure, a quantum causality measure, and an entanglement measure.

\item The quantum signalling measure.
\begin{align}
\mu_{\text{qsg}}^{A\rightarrow B}(G)=
\begin{cases}
1, \quad \text{A can quantum signal to B}
\\
0, \quad \text{A cannot quantum signal to B}.
\end{cases}
\end{align}
This function clearly obeys axioms 1 to 3 and is a quantum causality measure. It is also a normalized measure.

\item For quantum channels the quantum channel capacities are quantum causality measures, as one can easily check. Their normalization as in (\ref{eq:ncc}) are normalized quantum causality measures that assign the value one to noiseless channels.

\item For arbitrary correlations that may or may not contain indefinite causal structure, the one-shot entanglement transmission and subspace transmission capacities defined and studied in previous sections are quantum causality measures. Axioms 1 to 3 hold for these capacities directly by their definitions.

Definitions similar to (\ref{eq:ncm}) yield normalized capacities that assign the value one to the maximally causal correlations such as the identity channel:
\begin{align}
Q_{\text{norm}}(G;\epsilon):=\frac{Q(G;\epsilon)}{\sup_{G'\in \mathfrak{C}(G)}Q(G';\epsilon)},
\end{align}
where $\mathfrak{C}(G)$ is the set of correlations on the same systems of $G$, and $Q$ stands for any of the one-shot quantum capacities.

\end{itemize}

\section{Discussion}\label{sec:dis}

The present work is partially inspired by the previous work of Fitzsimons Jones and Vedral (FJV) who introduced ``causality monotones'' for ``pseudo-density matrices'' as a measure of causality \cite{fitzsimons2015quantum}. Pseudo-density matrices as they stand in \cite{fitzsimons2015quantum} are generalization of density matrices and describe qubit systems that reside at different times. Given a pseudo-density matrix $R$, a causality monotone $f(R)$ is required to obey
\begin{enumerate}
\item $f(R) \ge 0$, with $f(R) = 0$ if $R$ is completely positive, and $f(R) = 1$ for any $R$ obtained from two consecutive measurements on a single qubit closed system.
\item $f(R)$ is invariant under unitary operations.
\item $f(R)$ is non-increasing under local operations.
\item $\sum_i p_i f(R_i) \ge f(\sum_i p_iR_i)$.
\end{enumerate}

There are some obvious similarities and important differences between the FJV axioms for causality monotones and the axioms for (quantum) causality measures. The biggest difference is that the (quantum) causality measure axioms apply to general models, whereas the causality monotone axioms apply specifically to pseudo-density matrices. Pseudo-density matrices have some limitations which makes other frameworks more preferable to study quantum theory with generalized states. In particular, in more general models such as the process matrices systems can have arbitrary finite dimensions, measurements are not restricted to projective ones, and measurement update rules are more flexible.

In terms of the content of the axioms, the $f(R)\ge 0$ part of the first axiom of FJV is the same as Axiom 2 for (quantum) causality measures. The rest of FJV Axiom 1 depend on the specific construction of pseudo-density matrices and do not have general applicability. For general correlations Axioms 2 of FJV would have an analogue saying that $\mu^{A\rightarrow B}(G)$ is invariant under local unitary operations within $A$ and within $B$. Yet this would actually follows from causality measure Axiom 2 because local unitary operations are reversible. Suppose for some local unitary $U$, $\mu^{A\rightarrow B}(U(G))<\mu^{A\rightarrow B}(G)$. Then by causality measure Axiom 2, $\mu^{A\rightarrow B}(G)=\mu^{A\rightarrow B}(U^{-1}\circ U(G))<\mu^{A\rightarrow B}(U(G))<\mu^{A\rightarrow B}(G)$, which is a contradiction. By causality measure Axiom 2, $U$ cannot increase $\mu^{A\rightarrow B}(G)$ either. Hence any unitary operation must keep $\mu^{A\rightarrow B}(G)$ constant. Axiom 3 of FJV is the same as Axiom 1 for (quantum) causality measures. Axiom 4 of FJV expresses the convexity of causality measures and as already discussed in Subsection \ref{subsec:a} is too stringent because it would rule out communication capacities as causality measures.

Another significant difference is each particular (quantum) causality measure assigns two functions $\mu^{A\rightarrow B}$ and $\mu^{B\rightarrow A}$ to a pair of parties $A$ and $B$, while each particular causality monotone assigns only one function. With definite causal structure, there can only be causal influence for at most one direction, so one of the two measures $\mu^{A\rightarrow B}$ and $\mu^{B\rightarrow A}$ is constantly zero and is redundant. In this case it is reasonable to have one function as causal monotones do. However, when there is indefinite causal structure both $\mu^{A\rightarrow B}$ and $\mu^{B\rightarrow A}$ are relevant. 

As defined, the FJV causality monotones are not restricted to two parties, although the explicit examples studied in \cite{fitzsimons2015quantum} are bipartite pseudo-density matrices. The causality measures as defined in this work only apply to two parties. In the multipartite case it remains to be investigated what the FJV causality monotones measure, and how the causality measures generalize.

Janzing, Balduzzi, Grosse-Wentrup and Sch{\"o}kopf (JBGWS) proposed a set of five postulates for measures of causal strength \cite{janzing2013quantifying}. The authors mention that although they regard them as natural postulates, they ``do not claim that every reasonable measure of causal strength should satisfy these postulates''. In contrast, the axioms we propose in this paper are intended to hold for all measures of causality. A somewhat restrictive postulate of JBGWS is the ``mutual information'' postulate, which says that ``if the true causal DAG reads $X \rightarrow Y$, then $\mathfrak{C}_{X\rightarrow Y} =I(X;Y)$. Here $X$ and $Y$ are nodes of a DAG, $X \rightarrow Y$ means that $X$ influences $Y$ directly -- $X$ changes the probability distribution of $Y$ even if all other variables are held constant, $\mathfrak{C}_{X\rightarrow Y}$ is the causal strength measure for $X$ causally influencing $Y$, and $I(X;Y)$ is the classical mutual information. The mutual information postulate sets classical mutual information as the default causal strength measure when the condition of the postulate is met. In our view this postulate cannot be imposed on general causality measures because it excludes other useful causality measures such as the quantum capacities. Another restriction of the JBGWS postulates is that as stated they apply to causal models based on DAGs. As mentioned in Section \ref{sec:cm}, the causality measure axioms of this paper do not require the models to be based on DAGs, and hence apply to more general models such as the causaloid models \cite{hardy2005probability} which are naturally associated with hypergraphs rather than graphs. 

An advantage of the JBGWS postulates is that they take into consideration the possible variation of variables outside of the nodes associated with the causal arrows. In the present paper we assumed that the correlations $G$ are bipartite, so that there are no degree of freedom left in parties other than $A$ and $B$. Each $G$ is understood as a conditional correlation $G(v_C,v_D,\cdots)$ that arises when all other parties $C,D,\cdots$ have already conducted their operations and observed certain outcomes, with these parameters denoted as $v_C,v_D\cdots$ for $G$ to depend on. A more general study of causality measures that allows these parameters to vary would be a study of multipartite causality measures which also generalize the party $A$ exerting the causal influence and the party $B$ receiving the causal influence to multiple parties. This general study is left for future work.

\section{Conclusion}\label{sec:c}

In this paper we proposed three axioms for all reasonable quantitative measures of causal strength to obey. The axioms apply to any theory which contains the concepts of correlations that  mediate  the  causal  influence,  and  local  operations  that  change  the correlations in order to exert the causal influence. In particular, the axioms apply to theories with indefinite causal structure. The canonical examples of causality measures the various notions of information transmission capacities. For quantum models, we defined and studied the one-shot entanglement transmission and subspace transmission capacities as causality measures  in detail. These causality measures in turn allow us to define the notions of quantum signalling and quantum causality measures for quantum models such that correlations that cannot transmit quantum information have zero quantum causality measures.

Some natural questions arise from the present work that deserve further investigation. We studied causality measures for two parties. An obvious question is how to generalize the study to multiple parties. Another question concerns the use of one-shot quantum capacities to define quantum signalling. The motivation is to find the weakest criterion of quantum signalling that makes sense, and one-shot capacities which tolerate some errors yield a weaker criterion than capacities that do not tolerate any error (such as asymptotic capacities). Although it is reasonable to use quantum capacity to qualify quantum signalling and the one-shot quantum capacities give the weakest criterion among all standard quantum capacities, one needs to keep an open mind on finding still weaker criterion. In terms of the relation between causality measures and entanglement measures, we pointed out in Section \ref{sec:cm} that in frameworks where causality and entanglement measures can be defined, causality measures are special cases of entanglement measures that obey Axiom 3 in the LOCC setting without classical communication. Because of this connection it is possible to harness knowledge of entanglement measures and apply it to causality measures. Finally, we proved an interesting reconstruction theorem for the family of harmonic clean quantum models of indefinite causal structure that allows one to reconstruct the causally relevant parameters of the models from the one-shot entanglement transmission capacities. It is worth investigating further to what extent the one-shot capacities characterize general models of indefinite causal structure.

\section*{Acknowledgement}
I thank Lucien Hardy and Achim Kempf for guidance and support as supervisors, and Fabio Costa for making the crucial suggestion during the Spacetime and Information 2017 Workshop of using the communication capacity as a quantitative measure of causal structure.

Research at Perimeter Institute is supported by the Government of Canada through the Department of Innovation, Science and Economic Development Canada and by the Province of Ontario through the Ministry of Research, Innovation and Science. This work is supported by a grant from the John Templeton Foundation. The opinions expressed in this work are those of the author's and do not necessarily reflect the views of the John Templeton Foundation.

\bibliographystyle{apsrev}
\bibliography{bib.bib}

\end{document}